\documentclass[letterpaper, 10 pt, conference]{ieeeconf}  

\IEEEoverridecommandlockouts                              

\usepackage{graphicx}     
\usepackage{color}
\usepackage{courier}
\usepackage{algorithm}
\usepackage[noend]{algpseudocode}
\usepackage{amsmath}
\usepackage{amssymb}
\usepackage{amsfonts}
\usepackage{array}    
\usepackage{siunitx}

\usepackage{float}
\usepackage{tikz}
\usetikzlibrary{arrows,%
                plotmarks, calc}
\usetikzlibrary[positioning, calc, intersections, backgrounds]      

\usepackage{xspace}
\usepackage{soul}
\usepackage{mathtools}
\usepackage{makecell}
\usepackage{cite}
\usepackage{nicefrac}
\usepackage{float}
\usepackage{wrapfig}
\usepackage[utf8]{inputenc}
\usepackage{bbm}

\newcommand{\ra}{\rightarrow}

\newtheorem{asm}{Assumption}
\newtheorem{lem}{Lemma}
\newtheorem{thm}{Theorem}
\newtheorem{rem}{Remark}

\newcommand{\ie}{\unskip, i.\,e.,\xspace}

\newcommand{\R}{\ensuremath{\mathbb{R}}}

\newcommand{\X}{\ensuremath{\mathbb{X}}}

\newcommand{\U}{\ensuremath{\mathbb{U}}}

\newcommand{\sm}{\ensuremath{\setminus}}
\newcommand{\set}[1]{\ensuremath{\mathbb{#1}}}

\newcommand{\eps}{\ensuremath{\varepsilon}}

\newcommand{\ball}{\ensuremath{\mathcal B}}
\DeclareMathOperator*{\argmin}{arg\,min}

\newcommand{\norm}[1]{\left\lVert#1\right\rVert}

\newcommand{\abs}[1]{\left\lvert#1\right\rvert}
\newcommand{\scal}[1]{\left\langle#1\right\rangle}

\definecolor{dgreen}{rgb}{0.0, 0.5, 0.0}

\makeatletter
\newcommand{\subalign}[1]{%
	\vcenter{%
		\Let@ \restore@math@cr \default@tag
		\baselineskip\fontdimen10 \scriptfont\tw@
		\advance\baselineskip\fontdimen12 \scriptfont\tw@
		\lineskip\thr@@\fontdimen8 \scriptfont\thr@@
		\lineskiplimit\lineskip
		\ialign{\hfil$\m@th\scriptstyle##$&$\m@th\scriptstyle{}##$\crcr
			#1\crcr
		}%
	}
}

\title{\LARGE \bf Time-varying optimal control under measurement errors}

\author{Patrick Schmidt, Stefan Streif
        \thanks{The authors are with the Chemnitz University of Technology, Professorship Automatic Control and System Dynamics, 09107 Chemnitz
        {\tt\small $\{$patrick.schmidt, stefan.streif$\}$@etit.tu-chemnitz.de} \newline
        $\copyright$ 2026 the authors. This work has been accepted at the ECC for
publication under a Creative Commons Licence CC-BY-NC-ND.
        }
        }
\makeatletter
\newcommand*\bigcdot{\mathpalette\bigcdot@{.5}}
\newcommand*\bigcdot@[2]{\mathbin{\vcenter{\hbox{\scalebox{#2}{$\m@th#1\bullet$}}}}}
\makeatother

\begin{document}
\maketitle
\thispagestyle{empty}
\pagestyle{empty}

\begin{abstract}
	Solving optimal control problems to determine a stabilizing controller involves a significant computational effort.
	Time-varying optimal control provides a remedy by designing a tracking system, given as an ordinary differential equation, to track the solution of the optimal control problem. 
	To improve the applicability of the method, measurement errors are considered in this paper and it is described how these errors influence a control Lyapunov function-based decay condition.
	As a result of these investigations, input-affine constraints that meet the standard formulation and that describe the set of admissible controls are obtained. 
	The paper also derives a requirement on the necessary measurement accuracy as well as a triggering condition for taking a new measurement.
	The main theorem combines these results into a robustly stabilizing control algorithm, meaning that all closed-loop trajectories starting in a vicinity around the true state converge to zero.
	Additionally, the tracking system ensures that the optimal control is tracked at the end of each sampling period.
	The effectiveness of this approach is demonstrated using a train acceleration model and the well-known predator-prey model.	
\end{abstract}

\section{Introduction} \label{sec:intro}

Over the last few decades, many different methods have been developed to determine stabilizing controls.
They all have different performance characteristics, meaning some ensure fast convergence, while others ensure small control inputs.
To choose the best-performing method, the user can rank the different approaches using a user-defined performance or objective function.
An alternative to manually selecting the best-performing method is considering an optimal control approach, such as model predictive control (MPC) \cite{Boccia2014-stability}.
These approaches minimize a given objective function to determine the optimal control. 
However, optimization-based approaches have the disadvantage of an increased computational burden \cite{Grune2017-NMPC}.
Additionally, optimization inaccuracies (cf. \cite{Schmidt2021-inf}) might impair stabilization.

One approach that avoids the computational burden is time-varying optimal control (TVOC) \cite{Fazlyab2017-prediction, Simonetto2020-time-varying-survey}.
In TVOC, a tracking system given in the form of an ordinary differential equation is designed.
Its solution converges to the optimal solution of a relaxed objective function.
The latter is obtained by incorporating constraints via barrier functions into the objective function, which is strongly convex and considers the states and inputs at the current time instead of a finite or infinite time horizon.
The tracking system consists of a tracking part and a convergence part.
The latter can be freely designed, e.g., to obtain fixed-time convergence \cite{Polyakov2011-nonlinear-feedback}, which means that the optimal solution is obtained after a user-defined time.
Since the underlying optimization problem is not solved, nonexact optimization has no influence on the result.

This paper considers measurement errors which cause a discrepancy between the real output of the system and the measurement.
Since the error propagates to the estimated states, it is simply denoted as measurement error from now on.
Many works assume accurate measurements since most of the stabilization methods work if the errors are \emph{small enough}.
Additionally, \cite{Sontag1999-stability-disturb} claimed that for a continuous state-feedback control $\kappa(x)$, the behaviors of the two systems, one under $\kappa(x + e)$ with measurement error $e$, and the other under $\kappa(x)$ are close to each other.
In the context of TVOC, this work addresses the question of whether a measurement error is \emph{small enough} by deriving a sufficient bound on the measurement error to ensure stabilization.
Additionally, it derives a triggering rule that determines when new measurements need to be taken.
This approach guarantees the existence of a control that stabilizes all closed-loop trajectories starting in a neighborhood of the measured state, which includes also the true state.


The rest of the paper is structured as follows:
The problem setting and the procedure of the paper are described in Section \ref{sec:prob-setting}.
Section \ref{sec:method} recaps results from previous papers, like TVOC with its tracking system and it derives input-affine constraints that describe how measurement errors affect the set of admissible controls.
Additionally, the main theorem on stabilization in the presence of measurement errors is considered. 
The presented algorithm is applied to two examples in Section \ref{sec:case-study}, before the paper is concluded in Section \ref{sec:concl}.

Notation:
The scalar product of two vectors $x, y \in \R^n$ is denoted as $\scal{x,y}$. 
The Jacobian of a function $V$ is denoted as $\nabla_x V(x)$ and the Hessian as $\nabla_{xx} V(x)$. 
The set $\ball_{\eps}(\hat x)$ is defined as $\ball_{\eps}(\hat x) := \{ x \in \R^n: \norm{x - \hat x} \leq \eps \}$.
The closed-loop trajectory under a feedback $\kappa$ starting at $t = t_0$ in $x(t_0) = x_0$ is denoted as $x_\kappa(t; t_0, x_0)$.
Analogously, for a set of initial values, the set of closed-loop trajectories is defined as $x_\kappa(t; t_0, \ball_{\eps}(\hat x_0)) := \{ x_\kappa(t; t_0, x_0): x_0 \in \ball_{\eps}(\hat x_0) \}$.
Considering a set of states $\X \subset \R^n$, $\varphi(u,\X) \leq 0$ means that an inequality $\varphi(u,x) \leq 0$ has to hold for each $x \in \X$.

\section{Problem setting} \label{sec:prob-setting}

The input-affine system 
\begin{equation} \label{eqn:dynamics-f}
	\dot x = f(x) + g(x) u, \ x(0) = x_0
\end{equation}
with states $x \in \R^n$, controls $u \in \R^m$, and initial condition $x(0) = x_0$ is considered in the following sections.
The internal dynamics $f: \R^n \ra \R^n$ and the input coupling function $g: \R^n \ra \R^{n \times m}$ are assumed to be continuously differentiable in $x$.

The objective is to stabilize \eqref{eqn:dynamics-f} at an equilibrium point, without loss of generality at zero, where the control has to be optimal w.r.t. an objective function.
This task results in the following optimal control problem:
\begin{equation}
	\label{eqn:opt-prob-initial}
	\begin{split}
		\min_{u \in \U} \quad&J(u,x)  \\
		\mathrm{s.t.}\quad& \dot x = f(x) + g(x)u, \\
		&\scal{\nabla_x V(x), f(x) + g(x) u} \leq -w(x)
	\end{split}
\end{equation}
with initial value $x(0) = x_0$.
The set of input constraints $\U := [u_{1,\min}, u_{1, \max}] \times [u_{2,\min}, u_{2, \max}] \times  \ldots \times [u_{m,\min}, u_{m, \max}] \subset \R^m$ is assumed to be a compact set with $\{ 0 \} \in \U$.
The objective function $J: \R^m \times \R^n \ra \R$ is assumed to satisfy the following mild assumption (cf. \cite{Fazlyab2016-self-triggered}, \cite{Schmidt2021-tracking}, \cite{Simonetto2016-pred-corr}):
\begin{asm} \label{asm:J-twice-differentiable-J-uni-str-convex}
	The objective function $J(u,x)$ in \eqref{eqn:opt-prob-initial} is twice continuously differentiable in both arguments and is uniformly strongly convex in $u$ with parameter $m_J$ for all $x \in \R^n$ \ie $\nabla_{uu} J(u,x) - m_J I$ is positive semidefinite for all $u \in \R^m$ and all $x \in \R^n$.
	\hfill $\square$	
\end{asm}

The following assumption ensures the existence of a stabilizing control $u$ for \eqref{eqn:dynamics-f}.
\begin{asm} \label{asm:decay-condition}
	There exists a continuously differentiable feedback controller $\kappa(x)$ and a twice continuously differentiable radially unbounded Lyapunov function $V$ with $\alpha_1(x) \leq V(x) \leq \alpha_2(x)$ \cite{Khalil2002-nonlin-sys} such that for all $x \in \R^n$, it holds that 
	\begin{equation} \label{eqn:decay-condition}
		\dot V(x) = \scal{\nabla_x V(x),  f(x) + g(x) \kappa(x)} \leq - w(x), 
	\end{equation}
	where $w: \R^n \ra \R$ is continuously differentiable and satisfies $w(0) = 0$ and $x \not = 0 \implies w(x) > 0$.
		\hfill $\square$
\end{asm}
The inequality in \eqref{eqn:decay-condition} is written for an arbitrary $u$ as
\begin{equation} \label{eqn:varphi}
	\varphi(u,x) := \scal{\nabla_x V(x), f(x) + g(x) u} + w(x) \leq 0.
\end{equation}

Solving \eqref{eqn:opt-prob-initial} requires the true state, which can be determined by a measurement. 
Measuring the output commonly comes along with a measurement error which propagates through the state such that there exists a discrepancy between the true state and its measurement in the following way:
\begin{asm} \label{asm:measuring-procedure}
	For a true state $x_k := x(t_k)$ and its measurement $\hat x_k := \hat x(t_k)$, it holds that $\hat x_k \in \ball_\eps(x_k)$. 
	Since $x_k$ is unknown, also the center of the ball $\ball_\eps(x_k)$ is unknown.
	Therefore, $\ball_{2\eps}(\hat x_k)$ is considered since the true state $x_k$ is located in this set.
		\hfill $\square$
\end{asm}

Even if the measurement error is small, it can effect that two closed-loop trajectories starting at $\hat x_k$ and $x_k$ diverge under a control $\kappa(\hat x_k)$ \cite{Sontag1999-stability-disturb}.
At $t = t_k$, it is desired to determine a control $u_{k}(t_k)$ that stabilizes \eqref{eqn:dynamics-f} in the sense that 
\begin{equation} \label{eqn:robust-at-measurement}
	\varphi(u_{k}(t_k), \ball_{2\eps}(\hat x_k)) \leq 0
\end{equation}
holds, i.e., $u_{k}(t_k)$ robustly stabilizes the given system in the sense that all states $x_k$ that produce the measurement $\hat x_k$ are stabilized by the same input.
Additionally, $u_{k}$ has to be determined for a time interval, say $[t_k, t_{k+1})$, such that after applying $u_{k}(t)$ for $t \in [t_k, t_{k+1})$ all closed loop trajectories starting in $\ball_{2\eps}(\hat x_k)$ fulfill
\begin{equation} \label{eqn:robust-time-after-measurement}
	\forall t \in [t_k, t_{k+1}): \ \varphi(u_{k}(t), x_{u_{k}}(t; t_k, \ball_{2\eps}(\hat x_k))) \leq 0.
\end{equation}

The first question that arises is whether such a control $u_k$ can be always found.
The second question is when a new measurement has to be taken. 
These questions are answered further below in this paper.

Due to the presence of measurement errors, the goal of asymptotically stabilizing \eqref{eqn:dynamics-f} is given up in favor of practical stabilization, meaning that the closed-loop trajectories remain in some vicinity around the equilibrium point \cite{Schmidt2021-inf}.
This yields the following procedure for stabilization:
\newline
\textbf{Stabilization procedure and definition of balls:}
Based on the initial value $\hat x_0$, a starting ball $\ball_R$ with $R := \norm{\hat x_0} + 2\eps$ is defined.
A stabilizing control is determined such that the value of the control Lyapunov function (CLF) decreases for all future states, ensuring that all closed-loop trajectories remain within an overshoot ball $\ball_{R^\star}$ (see \eqref{eqn:robust-at-measurement} and \eqref{eqn:robust-time-after-measurement}).	 
The radius of this ball is given by $R^\star := \alpha_1^{-1}(\hat V)$, where $\hat V = \max_{\norm x \leq R} V(x)$ denotes the maximum value of the CLF within $\ball_R$. 
Stabilization is guaranteed until the closed-loop trajectories reach the core ball $\ball_{r^\star}$.
Within this region, an arbitrary control is applied as long as the closed-loop trajectories remain inside $\ball_r$, where $r^\star < \alpha_2^{-1}(\alpha_1(r)) - 2 \eps$ holds.
	
Details of this stabilization procedure are clarified in the next section, which presents time-varying optimal control as a method to solve \eqref{eqn:opt-prob-initial}.
Afterwards, the effect of the measurement error on the decay condition is presented.
As a last point in that section, the resulting relaxed optimal control problem including measurement errors is solved with TVOC.
		
\section{Method} \label{sec:method}

\subsection{Time-varying optimal control}

It is desired to solve the optimal control problem \eqref{eqn:opt-prob-initial} with methods from TVOC.
This technique is based on a tracking system 
\begin{equation} \label{eqn:tracking-system-generic}
	\dot u(t) = h(u,x,t), \ u(t_0) = \kappa(x_0) 
\end{equation}
that determines the solution of \eqref{eqn:opt-prob-initial} by simply solving the ODE \eqref{eqn:tracking-system-generic} instead of solving the OCP.
Assumption \ref{asm:J-twice-differentiable-J-uni-str-convex} ensures that such a tracking system can be designed in the sense that the tracking error $e_{\mathrm T}(t) := \norm{\tilde u(t) - u^\star(t)}$ is zero after a user-specified time $\tau > 0$.
The tracking error depends on the solution $\tilde u(t; t_0, \kappa(x_0))$ of \eqref{eqn:tracking-system-generic} as well as the solution of the relaxed optimal control problem $u^\star$, which is obtained in the following way:
The inequality constraint \eqref{eqn:varphi} is relaxed to $\varphi(u,x) \leq \gamma$ for some $\gamma > 0$, which has to hold for all $x \in \R^n$ and all $u \in \U$.
The relaxed inequality constraint is included in the objective function of \eqref{eqn:opt-prob-initial} via a barrier function (cf. \cite{Fazlyab2017-prediction}), i.e., a twice continuously differentiable convex function $B: \R_{\leq 0} \ra \R \cup \{ \infty \}$, where $B(z) < \infty$ holds for all $z < 0$ and $\lim\limits_{z \ra 0^-} B(z) \ra \infty$.
The input constraints $\U$ are also included in the objective function by defining the respective vector of input constraints with $\varphi_{\U,i}(u) = u_{i, \min} - u_i$ and $\varphi_{\U,m+i}(u) = u_i - u_{i, \max}$ for $i \in 1, 2, \ldots, m$.

Finally, this approach results in the following relaxed optimal control problem:
\begin{equation}
	\label{eqn:opt-prob-zwischen}
	\begin{split}
		\min_{u \in \R^m} \quad& \tilde J(u,x,t) \\
		\mathrm{s.t.}\quad& \dot x = f(x) + g(x) u,
	\end{split}
\end{equation}
with the relaxed objective function
\begin{equation} \label{eqn:dfn-of-J-tilde-zwischen}
	\tilde J(u,x,t) := J(u,x) + \mu(t)^\top B \left( \begin{pmatrix}
		\varphi_\U(u) \\
		\varphi(u,x)		
	\end{pmatrix} - \gamma_{2m+1} \right),
\end{equation}
where $\gamma_{2m+1} := \begin{pmatrix} \gamma & \ldots & \gamma \end{pmatrix}^\top \in \R^{2m+1}$.
Additionally, $\mu(t) := \begin{pmatrix} \mu_1(t) & \mu_2(t) & \ldots & \mu_{2m+1}(t) \end{pmatrix}^\top$ is a continuously differentiable function satisfying $\mu_i(t) > 0$ and $\dot \mu_i(t) < 0$ for all $i = 1, 2, \ldots, 2m+1$ and all $t \in \R_{\geq 0}$.
The barrier function in \eqref{eqn:dfn-of-J-tilde-zwischen} is also defined component-wise, i.e., $B(z) = \begin{pmatrix} B(z_1) & B(z_2) & \ldots & B(z_{2m+1}) \end{pmatrix}^\top$ for $z \in \R^{2m+1}$.

Since the objective function $\tilde J(u,x,t)$ is given as the sum of the strongly convex function $J(u,x)$ and the convex functions $\mu_i(t) B(\varphi_{\U,i}(u) - \gamma)$ and $\mu_{2m+1}(t) B(\varphi(u,x) - \gamma)$, it is again strongly convex with parameter $m_J$.
Therefore, \eqref{eqn:opt-prob-zwischen} has a unique solution for all $t \geq 0$ given as
\begin{equation} \label{eqn:opt-prob-zwischen-sol}
	\begin{split}
		u^\star(t) := \argmin_{u \in \R^m} \quad& \tilde J(u,x,t) \\
		\mathrm{s.t.}\quad& \dot x = f(x) + g(x) u
	\end{split}
\end{equation}
and it can be determined with the standard approach of TVOC, where a tracking system is used to track \eqref{eqn:opt-prob-zwischen-sol}.

\begin{lem}[Theorem 1 in \cite{Schmidt2021-tracking}]
	Consider system \eqref{eqn:dynamics-f} and optimization problem \eqref{eqn:opt-prob-zwischen}.
	Let Assumptions \ref{asm:J-twice-differentiable-J-uni-str-convex} and \ref{asm:decay-condition} hold and let $\tau$ be given.
	Define the tracking system
	\begin{equation} \label{eqn:special-choice-h}
		\begin{aligned}	
			\dot u(t) = &- \nabla_{uu}^{-1} \tilde J(u, x, t) \left[ \Psi(\nabla_u \tilde J(u, x, t); \tau) + \nabla_{ut} \tilde J(u, x, t) \right. \\
			&+ \left.\nabla_{ux} \tilde J(u, x, t)^\top (f(x) + g(x) u) \right] =: h(u,x,t; \tau), \\
			&u(0) = \kappa(x_0)
		\end{aligned}
	\end{equation}
	with $\Psi: \R^m \ra \R^m$ as
	\begin{equation} \label{eqn:G-fixed-time-functions-vector}
		\begin{aligned}
			&\Psi(u; \tau) = \begin{pmatrix}
				\psi(u_1; \tau) & \psi(u_2; \tau) & \cdots & \psi(u_m; \tau)
			\end{pmatrix}^\top, \\
			&\psi(u_i; \tau) = \frac \pi \tau (\abs{u_i}^{0.5} + \abs{u_i}^{1.5}) \text{sign}(u_i).
		\end{aligned}
	\end{equation}
	Tracking system \eqref{eqn:special-choice-h} ensures that the tracking error $e_{\mathrm T}(t) := \norm{\tilde u(t) - u^\star(t)}$ is zero for all $t \geq \tau$, where $\tau$ is known as the user-defined settling time.
	Additionally, $\tilde u$ is feasible such that $\varphi_{\U,i}(\tilde u(t)) \leq \gamma$ and $\varphi(\tilde u(t), x_{\tilde u}(t; t_0, x_0)) \leq \gamma$ holds for all $t \geq 0$.
	\hfill $\square$
\end{lem}

\begin{proof}
	See \cite[Theorem 1]{Schmidt2021-tracking}
\end{proof}
The input affine constraints $\varphi_\U(u) \leq 0$ and $\varphi(u,x) \leq 0$ in \eqref{eqn:dfn-of-J-tilde-zwischen} describe a set of admissible controls.
The next section explains how a measurement error changes this set.
Once this change is expressed through input-affine constraints, they can also be incorporated into the objective function with the aim to apply TVOC to the resulting problem.

\subsection{Effect of non-exact measurement on input constraints}

Choosing an input to fulfill \eqref{eqn:robust-at-measurement} and \eqref{eqn:robust-time-after-measurement} ensures that the closed-loop trajectories are located in the set $\set V := \ball_{\hat R^\star}$, where $\ball_{\hat R^\star}$ is a compact set given in the definition of the balls in Section \ref{sec:prob-setting}.
Inequality \eqref{eqn:varphi} is input-affine such that it can be written as $\varphi(u,x) = \beta_0(x) + \sum_{i = 1}^m \beta_i(x) u_i$.
All coefficients $\beta_i$ of $\beta(x)$ are Lipschitz continuous on a compact set, which is $\set V$ in this case (see \cite{Schmidt2024-some} for a detailed explanation).
Due to the Lipschnitzness of the coefficients, $\beta_i(\set V) := \{ \beta_i(x) \in \R: x \in \set V  \}$ is a compact interval for each $i = 0, 1, \ldots, m$ and thus, there exists both an upper and a lower bound for each set $\beta_i(\set V)$ given as $\beta_{i,\min}(\hat x) := \beta_i(\hat x) - L_i \eps$ and $\beta_{i,\max}(\hat x) := \beta_i(\hat x) + L_i \eps$.
Therefore, $\beta(x) \in [\beta_{0,\min}(\hat x), \beta_{0,\max}(\hat x)] \times [\beta_{1,\min}(\hat x), \beta_{1,\max}(\hat x)] \times \ldots \times [\beta_{m,\min}(\hat x), \beta_{m,\max}(\hat x)]$ and the infinitely many inequalities $\varphi(u,\ball_{2\eps}(\hat x)) \leq 0$ reduce to $2^{m+1}$ given as 
\begin{equation} \label{eqn:varphi-beta-min-max}
	\begin{split}
		\hat \varphi_j(u, \hat x) = &\left( \beta_0(\hat x) + (-1)^{j} L_0 \eps \right) \\
		&+ \sum_{i = 1}^m \left( \beta_i(\hat x) + (-1)^{\lceil \frac{j}{2^i} \rceil} L_i \eps \right) u_i \leq 0
	\end{split}
\end{equation}
for $j = 1, \ldots, 2^{m+1}$, which are obtained by considering the different combinations of $\beta_{k,\min}(\hat x)$ and $\beta_{\ell,\max}(\hat x)$ for $k, \ell \in \{ 0, 1, \ldots, m \}$ with $k \not = \ell$.
The approach is summarized in Figure~\ref{fig:effect_of_meas_error}.

\begin{figure}
	\centering
	\includegraphics[trim={0cm 0.4cm 0.2cm 0.125cm},width = 0.47\textwidth]{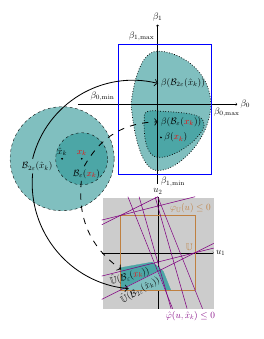}	
	\caption{Illustration of the approach: A measurement $\hat x_k$ is obtained, which is located in $\ball_\eps(x_k)$. 
	Since the unknown current state $x_k$ is unknown, the ball $\ball_{2\eps}(\hat x_k)$ is considered, which contains $\ball_\eps(x_k)$.
	Bounds for $\beta(x_k)$ are derived by considering $\beta(\ball_{2\eps}(\hat x_k))$.
	These bounds can be used to derive the set of admissible controls $\hat{\U}(\ball_{2\eps}(\hat x_k))$.
	Choosing some control that is located in this set ensures a decay for all $x \in \ball_{2\eps}(\hat x_k)$ and thus, also for the true state $x_k$.
	Since $\ball_{\eps}(x_k) \subset \ball_{2\eps}(\hat x_k)$, it holds that $\beta(\ball_{\eps}(x_k)) \subset \beta(\ball_{2\eps}(\hat x_k))$ and also $\hat{\U}(\ball_{2\eps}(\hat x_k)) \subset \hat{\U}(\ball_{\eps}(x_k))$.
	The set of admissible controls is required to incorporate measurement error in time-varying optimal control.}
	\label{fig:effect_of_meas_error}
\end{figure}

\subsection{Determining a triggering condition}

Referring to Section \ref{sec:prob-setting}, two questions need to be answered.
First, if there exists always a control $u_k$ such that \eqref{eqn:robust-time-after-measurement} is fulfilled and second, when a new measurement is required.

The first question is answered in \cite[Theorem 2]{Schmidt2024-some} and provides a statement when a measurement error is \emph{small enough}, meaning that the set of admissible controls (see Figure~\ref{fig:effect_of_meas_error}) is non-empty at each measurement.
That work derives a requirement on the measurement error, which is assumed to hold in this paper:

\begin{asm} \label{asm:eps-small-enough}
	The measurement error $\eps$ satisfies $\eps < \eps_{\min}$ with 
	\begin{equation} \label{eqn:lower-bound-bar-eps}
		\begin{split}
			&\eps_{\min} = \frac 1 2 \frac{\bar w}{L_0 + \sum \limits_{i = 1}^m L_i u_{i, \text M}}, \\
			&\bar w := \min_{x \in \set V \sm \ball_{r^\star}} w(x) - \tilde w(x), \\
			&u_{i, \text M} := \max\{ \abs{u_{i, \min}}, u_{i, \max} \},
		\end{split}
	\end{equation}	
	where $\tilde w$ is a relaxation of $w$, e.g., $\tilde w = \alpha w$, $\alpha \in (0,1)$.
		\hfill $\square$
\end{asm}
Both the relaxation of the decay as well as the radius of the core ball $r^\star$ can be determined by the user.
The existence of a robustly stabilizing control highly depends on the given $\eps$ since the measurement error has an influence on the $\beta_{i,\min}(\hat x)$ and $\beta_{i,\max}(\hat x)$.
Especially for measured $\hat x$ close to the origin, even a small $\eps$ has a large influence on the set of admissible controls.
Therefore, the vicinity around zero is excluded by defining a core ball $\ball_{r^\star}$ as a part of the target ball $\ball_r$ (cf. definition of the balls in Section \ref{sec:prob-setting}).
If the trajectory reaches this core ball, no decay can be guaranteed. 
Together with \eqref{eqn:varphi-beta-min-max}, it becomes clear that the smaller the required $\eps$ is, the smaller $r^\star$ has to be.
Additionally, a relaxation of the decay helps to enlarge both, $\eps$ and $r^\star$.
Further explanations of the relaxed decay condition as well as some additional details and explanations are given in \cite{Schmidt2024-some}.
  
The bound in Assumption \ref{asm:eps-small-enough} results from a computation of the maximum measurement error at each point $x \in \set V$ given by the following lemma.
\begin{lem}[Theorem 1 in \cite{Schmidt2024-some}] \label{thm:choice-of-eps} 
	Consider system \eqref{eqn:dynamics-f} together with input box constraints $\U$, where $\{ 0 \} \in \U$.
	Let Assumptions \ref{asm:decay-condition} and \ref{asm:eps-small-enough} hold and let $\tilde w: \R^n \ra \R_{\geq 0}$ be given as the relaxed decay rate, which is continuously differentiable and satisfies $\tilde w(x) \leq w(x)$ for all $x \in \R^n$, $\tilde w(0) = 0$ and $x \not = 0 \implies \tilde w(x) > 0$. 
	Then, 
	\begin{equation} \label{eqn:robustly-stabilizing}
		\begin{split}
			\forall \hat x \in \set V\sm \ball_{r^\star} \ &\exists \bar \eps(\hat x) > 2 \eps, \hat u \in \U: \varphi(\hat u, \ball_{\bar \eps(\hat x)} (\hat x); \tilde w) \leq 0
		\end{split}
	\end{equation}
	holds, where $\varphi(\hat u, x; \tilde w)$ denotes that $\tilde w$ is used instead of $w$ in $\varphi(\hat u, x)$.
	\hfill $\square$
\end{lem}
For an ease of notation, $\varphi(\hat u, x)$ is used as long as it is specified which decay is chosen.
Note that based on the definition of $\eps_{\max}$ in \eqref{eqn:lower-bound-bar-eps}, it holds that $\eps_{\max} \leq \bar \eps(\hat x)$ for all $\hat x \in \set V \sm \ball_{r^\star}$.

\begin{rem}	
	The required measurement accuracy at each point $\bar \eps(\hat x)$ in Lemma \ref{thm:choice-of-eps} was computed in \cite[Theorem 1]{Schmidt2024-some}.
	The cases for $m = 1$ and $m = 2$ are given in the appendix.
\end{rem}

Lemma \ref{thm:choice-of-eps} derives a bound on the maximum measurement error $\bar \eps(\hat x_k)$ for a given $\hat x_k$ to ensure robust stabilization in the sense of \eqref{eqn:robust-at-measurement}.
Since $\eps < \bar \eps(\hat x)$ for all $\hat x \in \set V \sm \ball_{r^\star}$, there exists a time interval of length $\delta_k := t_{k+1} - t_k$ such that the closed-loop trajectories can evolve until the set of admissible controls is empty (see \eqref{eqn:robust-time-after-measurement}).
Since both $\delta_{k}$ and the control $u_k$ are unknown, it is difficult to ensure \eqref{eqn:robust-time-after-measurement} without knowing the closed-loop trajectories.
The derived $\bar \eps(\hat x_k)$ yields a possibility to replace this condition, namely by ensuring $2\eps + \delta_k \bar F \leq \bar \eps(\hat x_k)$ with $\bar F$ as an upper bound on the dynamics given as
\begin{equation}
	\bar F := \sup_{x \in \set V \atop u \in \U} \norm{f(x) + g(x) u}.
\end{equation}	
The length of the time interval is given as $\delta_k := \frac{\bar \eps(\hat x_k) - 2 \eps}{\bar F}$.
Knowing $\bar \eps(\hat x_k)$ also allows to replace \eqref{eqn:robust-time-after-measurement} by $\varphi(u_{k}(t), \ball_{2\eps + \delta_k \bar F}(\hat x_k)) \leq 0$, which is more conservative than \eqref{eqn:robust-time-after-measurement} but easier to evaluate.

\subsection{Relaxed Optimal control problem}

Since the two questions at the beginning of Section \ref{sec:prob-setting} have been answered, a tracking system can be applied to the resulting optimal control problem (OCP) at $t = t_k$ to derive the optimal control $u_k^\star(t)$ for $t \in [t_k, t_{k+1})$.
This OCP reads as
\begin{equation}
	\label{eqn:opt-prob-finally}
	\begin{split}
		\min_{u \in \R^m} \quad& \tilde J(u,\hat x,t) \\
		\mathrm{s.t.}\quad& \dot{\hat x} = f(\hat x) + g(\hat x) u,
	\end{split}
\end{equation}
with the relaxed objective function
\begin{equation} \label{eqn:dfn-of-J-tilde-finally}
	\tilde J(u,\hat x,t) := J(u,\hat x) + \mu(t)^\top B \left( \begin{pmatrix}
		\hat \varphi(u, \hat x) \\
		\varphi_\U(u)
	\end{pmatrix} - \gamma_{2^{m+1}+2m} \right).
\end{equation}
The inequality $\varphi(u, \hat x) \leq 0$ can be neglected since it is fulfilled when $\hat \varphi(u, \hat x) \leq 0$ holds.
Objective function \eqref{eqn:dfn-of-J-tilde-finally} contains the measured states $\hat x$ instead of $x$, which are unknown.
However, based on the described stabilization procedure in Section \ref{sec:prob-setting}, a decay for all states in a vicinity around the measured value is ensured.
Thus, the true state $x$ remains close to $\hat x$ such that by minimizing $\hat x$, the minimum of $J(u, \hat x, t)$ is close to the unknown optimal value of $J(u, x, t)$.

The following lemma proposes a tracking system to determine the solution of \eqref{eqn:opt-prob-finally}.

\begin{lem} \label{thm:TVOC-under-ME-k}
	Consider the optimal control problem \eqref{eqn:opt-prob-finally} with measured initial value $\hat x_k \in \set V\sm\ball_{r^\star}$ and its solution $u_k^\star(t)$.
	Let Assumptions \ref{asm:J-twice-differentiable-J-uni-str-convex}, \ref{asm:decay-condition}, and \ref{asm:eps-small-enough} hold.
	Let the time until the next measurement time point be given as	
	\begin{equation} \label{eqn:triggering-time-points-nur-k}
		\delta_k =  \frac{\bar \eps(\hat x_k) - 2 \eps}{\bar F}	
	\end{equation}
	and define the solution of tracking system (cf. \eqref{eqn:special-choice-h})
	\begin{equation} \label{eqn:special-choice-h-delta}
		\begin{aligned}	
			\dot u(t) = h(u,\hat x,t; \delta_k), \ u(t_k) \in \hat{\U}(\hat x_k)
		\end{aligned}
	\end{equation}		
	as $\tilde u_k(t) := \tilde u_k(t; t_k, u(t_k))$ with $\hat{\U}(\hat x_k) := \{ u \in \U: \hat \varphi(u, \hat x_k) \leq 0 \}$.
	Applying the solution $\tilde u_k(t)$ to \eqref{eqn:dynamics-f} for $t \in [t_k, t_{k+1}), t_{k+1} := t_k + \delta_k$ ensures that the tracking error vanishes at the end of the settling time, i.e. $e_{\mathrm T}(t) := \norm{\tilde u_k(t_{k+1}) - u_k^\star(t_{k+1})} = 0$, and that system \eqref{eqn:dynamics-f} is robustly stabilized in the sense that 
	\begin{equation} \label{eqn:robust-time-after-measurement-u-tilde}
		\forall t \in [t_k, t_{k+1}): \ \varphi(\tilde u_{k}(t), x_{\tilde u_{k}}(t; t_k, \ball_{2\eps}(\hat x_k))) \leq 0.
	\end{equation}
	holds.
	\hfill $\square$
\end{lem}
\begin{rem}
	The choice $u(t_k) \in \hat{\U}(\hat x_k)$ of the initial value of the tracking system ensures feasibility at $t = t_k$, which is not necessarily fulfilled with $\kappa(\hat x_k)$.
\end{rem}
\begin{proof}
	There exist only two small differences between Lemma \ref{thm:TVOC-under-ME-k} and \cite[Theorem 1]{Schmidt2021-tracking}.
	First, the settling time is chosen as $\tau = \delta_k$ in \eqref{eqn:special-choice-h-delta}, which ensures that the tracking error vanishes after $\delta_k$ seconds.
	Second, the relaxed objective function \eqref{eqn:dfn-of-J-tilde-finally} contains the derived inequalities $\hat \varphi(u, \hat x)$ and $\varphi_\U(u)$.
	Since these inequalities are input-affine, $\tilde J$ is again strongly convex. 
	Additionally, \cite[Theorem 1]{Schmidt2021-tracking} does not depend on a special choice of the relaxed objective function.
	Thus, the proof can be also applied to Lemma \ref{thm:TVOC-under-ME-k}.
	Feasibility of the solution is ensured by the choice of $\delta_k$ to ensure $2\eps + \delta_k \bar F \leq \bar \eps(\hat x_k)$, since $\varphi(\tilde u_k(t), \ball_{\bar \eps(\hat x_k)}(\hat x_k)) \leq \gamma$ implies $\varphi(\tilde u_k(t), x_{u_k}(t; t_k, \ball_{2\eps}(\hat x_k))) \leq \gamma$ for all $t \in [t_k, t_{k+1})$, which was discussed at the end of the last subsection.
	Thus, robust stabilization is given.
\end{proof}
The results of the last sections are used to define the main algorithm.

\subsection{TVOC under measurement errors}

According to Lemma \ref{thm:TVOC-under-ME-k}, the proposed tracking system \eqref{eqn:special-choice-h-delta} ensures a vanishing tracking error after $\delta_k$ seconds. 
Additionally, a decay of all closed-loop trajectories starting in $\ball_{2\eps}(\hat x_k)$ is ensured if $\hat x_k \in \set V \sm \ball_{r^\star}$.
In the following theorem about time-varying optimal control in the presence of measurement errors, Lemma \ref{thm:TVOC-under-ME-k} is combined with the stabilization procedure in Section \ref{sec:prob-setting}.

\begin{thm} \label{thm:TVOC-under-ME}
	Consider the optimal control problem \eqref{eqn:opt-prob-finally} with its solution $u^\star(t)$ that consists of $u_k^\star(t)$ defined on $[t_k, t_{k+1})$ and let Assumptions \ref{asm:J-twice-differentiable-J-uni-str-convex}, \ref{asm:decay-condition}, and \ref{asm:eps-small-enough} hold.
	Let a relaxation of the decay rate be given as $\tilde w$ along with $r$ to obtain the balls in Section \ref{sec:prob-setting}, i.e.,
	\begin{itemize}
		\item Starting ball $\ball_R$, $R := \norm{\hat x_0} + 2\eps$,
		\item Overshoot ball $\ball_{R^\star}$, $R^\star := \alpha_1^{-1}\left( \max_{\norm x \leq R} V(x) \right)$,
		\item Target ball $\ball_r$,
		\item Core ball $\ball_{r^\star}$, $r^\star < \alpha_2^{-1}(\alpha_1(r)) - 2 \eps$.
	\end{itemize}
	Let the triggering rule for a new measurement be defined as
	\begin{equation} \label{eqn:triggering-time-points}
		\delta_k = \begin{cases} \frac{\bar \eps(\hat x_k) - 2 \eps}{\bar F}, & \hat x_k \not\in \ball_{r^\star} \\
		\frac{\alpha_2^{-1}(\alpha_1(r)) - 2 \eps - r^\star}{\bar F_0}, & \hat x_k \in \ball_{r^\star}
		\end{cases}
	\end{equation}
	with $\bar F_0 := \sup_{x \in \set V} \norm{f(x)}$.
	Define the solution of tracking system \eqref{eqn:special-choice-h-delta} as $\tilde u_k(t)$, $t  \in [t_k, t_{k+1})$.
	On the one hand, if the obtained measurement is outside the core ball, applying the solution $\tilde u_k(t)$ for $t \in [t_k, t_{k+1})$ to \eqref{eqn:dynamics-f} ensures that the tracking error vanishes at the end of the settling time and that system \eqref{eqn:dynamics-f} is robustly stabilized in the sense that \eqref{eqn:robust-time-after-measurement-u-tilde} holds until the core ball is reached. 
	On the other hand, if the obtained measurement is inside the core ball, applying $\tilde u_k(t) = 0$ for $t \in [t_k, t_{k+1})$ ensures that the trajectory remains in the target ball.
	\hfill $\square$
\end{thm}

\begin{rem}
	Once the core ball is reached, an arbitrary control, e.g. $u = 0$, can be applied to the system, since there exists not necessarily a stabilizing control.
	Instead of choosing $u = 0$ in $\ball_{r^\star}$, one might, for example, maximize the so-called inter-execution time $\delta_k = t_{k+1} - t_k$ by choosing  
	\begin{equation}
		u_\star := \argmin_{u \in \U} \sup_{x \in \set V} \norm{f(x) + g(x) u}
	\end{equation}	
	to minimize the number of measurements.
	Depending of the definition of the objective function, however, $u = 0$ might yield the smallest value.
\end{rem}

\begin{proof}
	The proof is divided in three parts: The first two parts consider the two cases whether the measurement is located in the core ball or not. The last part combines all time intervals to a solution of the optimal control problem.
			
	\underline{Part 1:} Outside the core ball: \newline
	Due to Lemma \ref{thm:TVOC-under-ME-k}, the decay is ensured based on feasibility of the solution outside the core ball.
	
	\underline{Part 2:} Inside the core ball: \newline
	For a measurement inside the core ball, the sampling time $\delta_k$ has to be chosen such that a new measurement is obtained if $\norm{\hat x_{\tilde u_k}(t; t_k, \ball_{\bar \eps(\hat x_k)}(\hat x_k))} \geq \tilde r := \alpha_2^{-1}(\alpha_1(r))$.
	The respective ball $\ball_{\tilde r - \eps}$ with $\tilde r \in (r^\star, r)$ is called triggering ball.
	It ensures that after a measurement, the closed-loop trajectories remain in the target ball even if they run along the border of the same level set of the CLF (cf. \cite[Figure 4]{Schmidt2024-some}).	
	Along with $\norm{\hat x_k} \leq r^\star \implies \norm{x_k} \leq r^\star + \eps$, it follows that the distance between the core ball and the triggering ball is given as $2 \eps + \delta_k \bar F_0$.
	Based on the choice of the triggering rule \eqref{eqn:triggering-time-points}, i.e., $\delta_k = \frac{\tilde r - 2\eps - r^\star}{\bar F_0}$, the desired result follows.	
	In this case, $u_k^\star := 0$ and $\tilde u_k := 0$ are set since no OCP is given and thus, no tracking system is applied.
	
	\underline{Part 3:} Combination of the procedure: \newline
	Now it is shown that the procedure yields a control that robustly stabilizes the given system by combining the obtained solutions $\tilde u_k$.
	At $t = t_k$, a new measurement $\hat x_k$ is determined. 
	
	Case 1 ($\hat x_k \not\in \ball_{r^\star}$): The inter-execution time $\delta_k$ is computed and a control $\tilde u_k: [t_k, t_{k+1}) \rightarrow \R^m$ is determined as a solution of the tracking system.
	The derived control robustly stabilizes the system in the sense of \eqref{eqn:robust-time-after-measurement-u-tilde} (see Part 1).
	After a measurement at $t = t_{k+1}$, there exist two subcases:
	One could either be located outside the core ball, i.e. $\hat x_{k+1} \not\in \ball_{r^\star}$, which is exactly the same procedure as it was described in Case 1, or a measurement $\hat x_{k+1} \in \ball_{r^\star}$ is obtained, which is given in Case 2.
		
	Case 2 ($\hat x_k \in \ball_{r^\star}$): An arbitrary control is applied such that the closed-loop trajectories $x_0(t_{k+1}; t_k, \ball_{2\eps}(\hat x_k))$ still remain in the triggering ball (see Part 2).
	After a new measurement at $t = t_{k+1}$, there exist, again, two subcases: 
	If $\hat x_{k+1} \in \ball_{r^\star}$, $\tilde u_{k+1} := 0$ is again applied to the system and by definition of $\delta_{k+1}$, a new measurement is obtained before the closed-loop trajectories leave the triggering ball (see Case 2).
	If $\hat x_{k+1} \not\in \ball_{r^\star}$, it holds that $\hat x_{k+1} \in \ball_{\tilde r}$ and thus, the level set $\{ x \in \R^n: V(x) \leq c := \max_{x \in \ball_{\tilde r}} V(x) \}$ is contained in the target ball $\ball_{r}$ and Case 1 appears.
	
	Applying Case 1 until Case 2 appears ensures a convergence to the target ball. 
	The definition of $\delta_k$ and the triggering ball in Case 2 ensure that the closed-loop trajectories remain in the target ball once they entered it.
\end{proof}

Algorithm \ref{alg} summarizes the procedure of Theorem \ref{thm:TVOC-under-ME}.

\begin{algorithm}[H]
	\caption{Stabilization under measurement errors}
	\label{alg}
	\begin{algorithmic}[1]
		\renewcommand{\algorithmicrequire}{\textbf{Input:}}
		\renewcommand{\algorithmicensure}{\textbf{Compute:}}
		\Require System $\dot x = f(x) + g(x) u$, input constraints $\U$, CLF $V$, decay rate $w$ with relaxation $\tilde w$, measurement accuracy $\eps$ (according to Asm. \ref{asm:eps-small-enough}), target ball radius $r$
		\State Measure the initial state $\hat x_0$
		\State Compute radii $R$, $R^\star$, $r^\star$ of the balls based on given $r$
		\While {system is running}
			\State Measure the state to obtain $\hat x_k$
			\If {$\norm{\hat x_k} > r^\star$}
				\State Compute $\delta_k$ based on triggering rule \eqref{eqn:triggering-time-points}
				\State Compute $\beta_i$ and $L_i$, $i = 0, 1, \ldots, m$ on $\set V$
				\State Update $\hat \varphi(u, \hat x)$ in the OCP \eqref{eqn:opt-prob-finally}
				\State Apply $\tilde u_k(t)$ given from \eqref{eqn:special-choice-h-delta} for $\delta_k$ seconds
			\Else
				\State Compute $\delta_k$ based on triggering rule \eqref{eqn:triggering-time-points}
				\State Apply arbitrary control, e.g., $\tilde u_k = 0$ for $\delta_k$ seconds
			\EndIf
		\EndWhile
	\end{algorithmic}
\end{algorithm}

\begin{rem}
	If a local bound of the inter-execution times $\delta_k$ can be determined, a fixed settling time in tracking system \eqref{eqn:special-choice-h-delta} can be chosen based on this bound.
\end{rem}

\section{Case study} \label{sec:case-study}

\subsection{Train acceleration model}

Consider a simplified train motion model \cite{Hauck2023-map, Schmidt2024-some}
\begin{equation}
	\label{eq:model_train}
	\dot x = - \frac{F_\text{res}(x)}{m} + \frac{F_\text{train}(x)}{m} u,
\end{equation}
resulting from a balance of force equation with $m$ as the mass of the train and $x$ as its velocity.
The velocity can be controlled using a driving lever, whose position is assumed to be bounded between $[-1,1]$, by convention, where a positive value results in acceleration and a negative one in braking.
The resistance force consists of air resistance, inclination resistance, and rolling-resistance and can be expressed as
\begin{equation}\label{eq:F_resistance}
	F_\text{res}(x) = p (x - v_\text{W})^2 + q
\end{equation}
with parameters $p, q > 0$.
The train-specific driving power is given as
\begin{equation}
	F_\text{train}(x) =  k_1 e^{- k_2 x} + k_3
\end{equation}
with $k_i > 0$.
The parameters are given as $p = \SI{5.18}{\kilogram / \meter}$, $q = \SI{13046.32}{\newton}$, $m = \SI{68200}{\kilogram}$, $k_1 = \SI{1.516e+05}{\kilogram \meter \per \second^2}$, $k_2 = \SI{0.1147}{1 \per \second}$, $k_3 = \SI{1.564e+04}{\kilogram \meter \per \second^2}$.

It is desired to accelerate the train from $\hat x_0 = \SI{27}{\meter / \second}$ to $x^\star = \SI{30}{\meter / \second}$ at an inclination of one degree with wind speed $v_\text{W} = \SI{5}{\meter / \second}$ while minimizing the driving lever position, i.e., $J(u,x) = \frac 1 2 u^2$.

A crucial part is the weighting of the inequalities in \eqref{eqn:dfn-of-J-tilde-finally}.
Therefore, the relaxed objective function is replaced by
\begin{equation} \label{eqn:dfn-of-J-tilde-example}
	\tilde J(u,\hat x,t) := J(u,\hat x) + \mu(t)^\top B \left( W \cdot \begin{pmatrix}
		\hat \varphi(u, \hat x) \\
		\varphi_\U(u)
	\end{pmatrix} - \gamma \right).
\end{equation}
with a weighting matrix $W = \text{diag}(w_1, w_2, \ldots, w_{2^{m+1}+2m})$.
These weights along with the functions $\mu_i$ have a large influence on the solution.
On the one hand, the term $B(u - u_{\max} - \gamma)$ punishes high $u$.
On the other hand, the decay conditions $\hat \varphi(u, \hat x) \leq 0$ require a high $u$ for a large decay and thus, a fast convergence to the desired velocity $x^\star$.
Additionally, the inequalities in the barrier functions are of different order of magnitude, such that their scaling is also a possibility to influence the decay.
In this example, these parameters were chosen as follows: $W = \text{diag}(3, 3, 3, 3, 1, 1)$ and $\mu_i(t) = \exp(- 0.5 \cdot t)$.
The parameter in the exponential function describes how long the inequalities have to be considered.

Using the quadratic CLF $V(x) = \frac 1 2 x^2$ and the stabilizing controller
\begin{equation}
	\kappa(x) = - \text{tanh}(x - \text{arctanh}( u^\star ) - x^\star)
\end{equation}
with $u^\star = \frac{ F_\text{res}(x^\star)}{F_\text{train}(x^\star)}$
results in the decay
\begin{equation}
	\dot V = \scal{\nabla_x V, f + g \kappa} \leq - 0.025(x - x^\star)^2 =: - w(x)
\end{equation}
for $x < 40$, i.e., for velocities up to $\SI{40}{\meter / \second}$, which is the physical bound for the given train.
The relaxed decay is chosen as $\tilde w = 0.015(x - x^\star)^2$.
The measurement error is chosen as $\eps = 0.01$ and the radii of the balls are given as $r = 1$ (target ball), $\tilde r = 0.7$ (triggering ball), and $r^\star = 0.5$ (core ball).
Additionally, $\gamma = 0.01$ and $B(z) = -z^{-1}$ are chosen in the relaxed objective function.

The convergence of the closed-loop trajectories with the initial state $x_0$ and its measurement $\hat x_0$ is shown in Figure~\ref{fig:case_study_1_results}.
The lower picture shows that both trajectories remain in the magenta target ball after entering it once, which is ensured by the triggering condition.
It switches between a stabilizing control given from the tracking system \eqref{eqn:special-choice-h-delta} until the cyan core ball is reached, and an arbitrary control (here $u = 0$) which is applied as long as the trajectories are located in the triggering ball.

\begin{figure}
	\centering
	\includegraphics[trim={0.5cm 0.5cm 0cm 0cm},width = 0.47\textwidth]{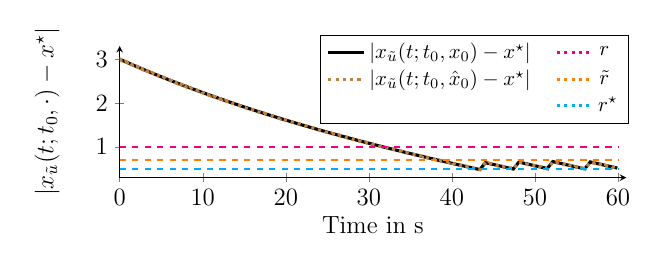}	
	\includegraphics[trim={0.5cm 0.5cm 0cm 0cm},width = 0.47\textwidth]{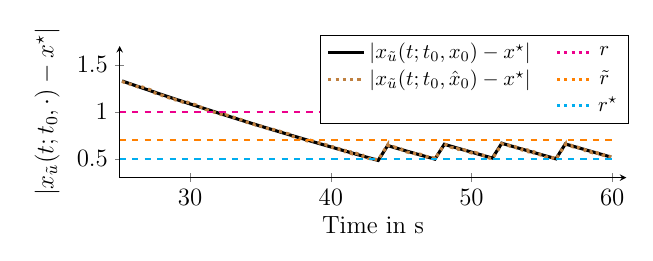}
	\caption{Upper picture: Convergence of the closed-loop trajectories into the target ball.	
	Lower picture: Closer look of the upper figure for $t \in [25,60]$.}
	\label{fig:case_study_1_results}
\end{figure}

Using the proposed triggering condition, a Zeno behavior, i.e. an infinite number of events that occur in a finite-time interval, can be excluded by considering the lower bounds of the terms included in \eqref{eqn:triggering-time-points}.
In the case of a measurement outside the core ball, there exists a lower bound for $\bar \eps(\hat x_k)$ given as $\eps_{\max}$ in \eqref{eqn:lower-bound-bar-eps}.
Since this bound is quite conservative, $\bar \eps(\hat x_k) - 2 \eps$ is not close to zero.
Thus, the same holds also for $\delta_k$.
In the case of a measurement inside the core ball, the lower bound for $\delta_k$ depends on the choice of $r^\star$, such that a Zeno behavior can be avoided by a suitable choice of the radius of the core ball as well as the two functions $\alpha_1$ and $\alpha_2$.

Based on \eqref{eqn:lower-bound-bar-eps}, the upper bound for the measurement error is computed as $\eps_{\max} = 1.7 \cdot 10^{-3}$.
Considering the theoretical results, stabilization is ensured for all measurement errors that are smaller than than $\eps_{\max}$.
Since this bound is only a sufficient condition, the measurement error was chosen as $\eps = 10^{-2}$ in this example.

Beside the the two functions $\alpha_1$ and $\alpha_2$, also other parameters can be changed that have a significant influence on the result.
Here, $\tilde w = 0.6 w$ and $r^\star = 0.5$ were chosen, which results in $\eps_{\max} = 1.7 \cdot 10^{-3}$.
If the  radii of the triggering ball and the core ball are halved, the sufficient bound for the measurement error is given as $\eps_{\max} = 0.43 \cdot 10^{-3}$, which reduces the maximum measurement error significantly.
Reducing these radii is, for example, required if a smaller radius of the target ball is desired, meaning that the deviation of the desired velocity of the train is smaller.
Note that $\eps_{\max}$ is not directly influenced by the choice of $r$, but only through $r^\star$, and thus indirectly by $r$.
As a possibility to increase the bound $\eps_{\max}$, the decay can be relaxed to $\tilde w = 0.3 w$.
It results in $\eps_{\max} = 3 \cdot 10^{-3}$, meaning that larger measurement errors are allowed since a smaller decay is required.
These values give only a small insight how the different parameters influence each other.

\subsection{Lotka Volterra model}

Consider the Lotka-Volterra predator-prey model as
\begin{equation} \label{eqn:Lotka-Volterra}
\frac{\mathrm d}{\mathrm d t}
\begin{pmatrix}
         x_1\\
         x_2 \\
       \end{pmatrix}=\underbrace{\begin{pmatrix}
         \bar \alpha x_1 - \bar \beta x_1 x_2 \\
         -\bar \gamma x_2 + \bar \delta x_1 x_2 \\
       \end{pmatrix}}_{= f(x)}+\underbrace{\begin{pmatrix}
         x_1 & 0\\
         0 & x_2 \\
       \end{pmatrix}}_{=g(x)}\begin{pmatrix}
         u_1\\
         u_2 \\
       \end{pmatrix},
\end{equation}
with additional controls $u_1, u_2$ as an input.
The model describes the interactions between a prey $x_1$ as the first species and a predator $x_2$ as the second species.
The parameters in \eqref{eqn:Lotka-Volterra} are chosen as $\bar \alpha=1.1$, $\bar \beta=0.4$, $\bar \gamma=0.4$, and $\bar \delta=0.1$ and they are known as the prey growth rate, prey death rate, predator death rate, and predator growth rate.
The aim is to stabilize \eqref{eqn:Lotka-Volterra} at $x^\star=\begin{pmatrix} 10 & 4 \end{pmatrix}^\top$.
The objective function is given as $J(u,x) = \frac 1 2 x^\top \begin{pmatrix} 1 & 0 \\ 0 & 3 \end{pmatrix} x$ meaning that it is more costly to reduce the amount of predators than prey.

A CLF for \eqref{eqn:Lotka-Volterra} is given as \cite{Meza2005-controller}: 
\begin{equation}
	V(x) =x_1 - x_1^\star - x_1^\star \log\left(\frac{x_1}{x_1^\star}\right) + x_2 - x_2^\star - x_2^\star \log\left(\frac{x_2}{x_2^\star}\right),
\end{equation}
which satisfies $V(x^\star) = 0$ and $V(x) > 0, \forall x \in \R_{> 0}^2 \sm \{ x^\star \}$.
The CLF allows to derive a feasible control according to Assumption \ref{asm:decay-condition} as
\begin{equation} \label{eqn:kappa-lotka-volterra}
		\kappa(x) = \begin{pmatrix}
			- \bar \alpha + \bar \beta x_2 + \tanh( - (x_1 - x_1^\star)) \\
			\bar \gamma - \bar \delta x_1 + \tanh( - (x_2 - x_2^\star))
		\end{pmatrix}.
\end{equation}
Using $\kappa(x)$, the decay rate is given as
\begin{equation} \label{eqn:derivative-CLF-Lotka-Volterra}
	\begin{aligned}
		\dot V(x) =& -\frac 1 2  (x_1 - x_1^\star) \tanh( - (x_1 - x_1^\star)) \\
		&+ \frac 1 2 (x_2 - x_2^\star) \tanh( - (x_2 - x_2^\star)) =: -w(x).
	\end{aligned}
\end{equation}
The relaxed decay rate is chosen as $\tilde w := \frac 1 2 w$.
In \cite{Schmidt2024-some}, bounds for the input are determined based on knowledge of the system as $\kappa(x) \in [-3, 4] \times [-3, 2]$, which yields $\U$.
The procedure is considered for a set of initial values $x_0 \in \set X_0$, namely 
\begin{equation}
		\set X_0 := \left\{ \begin{pmatrix} 5 \\ 8 \end{pmatrix}, \begin{pmatrix} 10 \\ 6 \end{pmatrix}, \begin{pmatrix} 15 \\ 4 \end{pmatrix}, \begin{pmatrix} 10 \\ 3 \end{pmatrix}, \begin{pmatrix} 5 \\ 2 \end{pmatrix}, \begin{pmatrix} 1 \\ 3 \end{pmatrix}, \begin{pmatrix} 1 \\ 5 \end{pmatrix} \right\}.
\end{equation}

The radii of the balls are given as $r = 0.7$ (target ball), $\tilde r = 0.3$ (triggering ball), and $r^\star = 0.2$ (core ball).
The measurement error is chosen as $\eps = 10^{-2}$.

Figure~\ref{fig:case_study_results} shows the closed-loop trajectories starting at the seven different initial values $x_0 \in \set X_0$.
All of these trajectories converge into the magenta target ball.
Stabilization using tracking system \eqref{eqn:special-choice-h-delta} is ensured as long as a measurement outside the cyan core ball is taken.
If a measurement inside the core ball is obtained, $u = \begin{pmatrix} 0 & 0 \end{pmatrix}^\top$ is applied.
It can be seen that the triggering condition in these cases is conservative, since the closed-loop trajectories do not diverge that fast to the target ball.
Additionally, the choice of the objective function can be seen in Figure~\ref{fig:case_study_results}.
Since it is more costly to remove predators $x_2$ from the population, the optimal strategy involves removing prey $x_1$ first.

\begin{figure}   
        \centering
       \hspace*{1cm}
       \includegraphics[trim={2cm 1.6cm 0.5cm 0cm},clip,width = 0.35\textwidth]{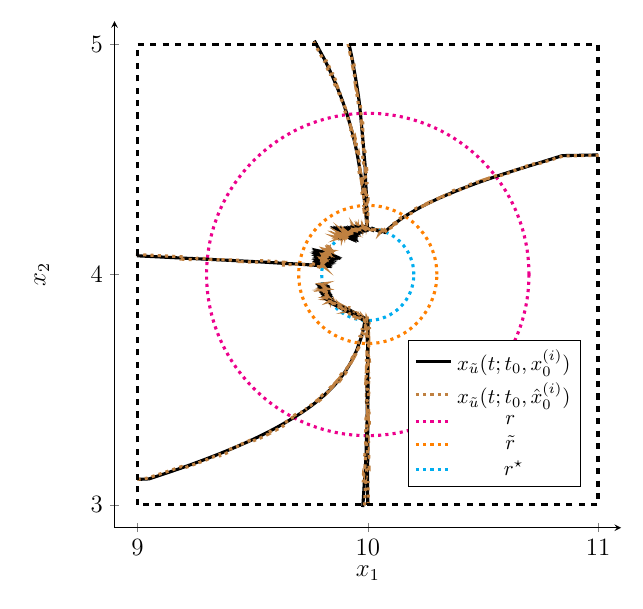}
       \includegraphics[trim={0cm 0.5cm 0cm 0.5cm},width = 0.47\textwidth]{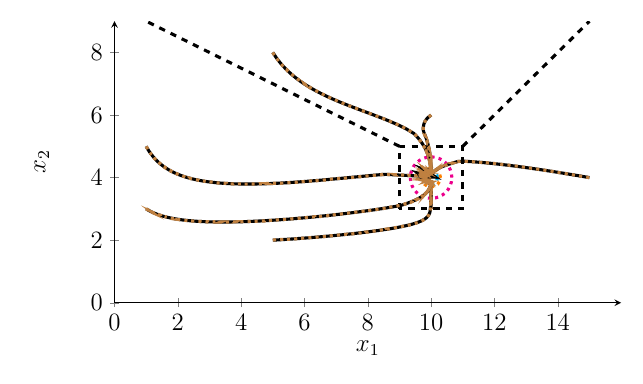}
        \caption{The closed-loop trajectories starting in the seven different initial values $x_0^{(i)} \in \set X_0$. All closed-loop trajectories starting in the real values (black solid lines) and in the measured values (brown dotted lines) remain in the target ball after entering it once.}
        \label{fig:case_study_results}
\end{figure}

\section{Conclusion} \label{sec:concl}

The paper discussed time-varying optimal control (TVOC) as a technique to enhance an existing stabilizing controller based on a user-defined objective function in case of measurement errors.
The paper described how measurement errors affect the set of admissible controls and how this set shrinks over time.
These results enabled the derivation of a triggering condition for the next required measurement, namely when no control exists that stabilizes all states located in a vicinity around the measured state.
A tracking system was derived that determines the optimal solution until the end of the sampling period without solving the optimal control problem.
The main theorem showed that the proposed algorithm practically robustly stabilizes the given input-affine system meaning that the closed-loop trajectories remain near the origin. 
The approach was applied to two systems in a case study.

\section*{Acknowledgment}
This research was funded by the German Federal Ministry of Research, Technology and Space (BMFTR) in the frame of the SRCC$\_$5D.Rail project, grant number 03WIR1225A.

\bibliographystyle{plain}
\bibliography{constructing-LFs,discont-DE,stabilization,stability,non-smooth-analysis,sliding-mode,MPC,opt-ctrl,dyn-sys,semiconcave,time-varying-optimization,triggered,viability-kernel,schmidt,PDE,nonlin-ctrl}

\appendix

\section{Required measurement accuracy}

The required measurement accuracy (see Lemma \ref{thm:choice-of-eps}) was computed in \cite{Schmidt2024-some} and is given as
\begin{equation}
	\bar \eps(x) = \begin{cases} 
		\bar \eps_0, & \text{ if } \beta_0(x) \leq 0 \land \beta_1(x) = 0 \\
		\bar \eps_1, & \text{ if } \beta_0(x) > 0 \land \beta_1(x) \not= 0 \\
		\min\{\bar \eps_0, \bar \eps_1\}, & \text{ else }
	\end{cases}
\end{equation}
with
\begin{equation}
	\bar \eps_0(x) = - \frac{\beta_0(x)}{L_0}
\end{equation}
and for $m = 1$
\begin{equation}
	\bar \eps_1 := \min \left\{ E_1, E_{01} \right\}
\end{equation}
with
\begin{equation}
	E_1 := \frac{\abs{\beta_1}}{L_1}, E_{01} := - \frac{\beta_0 + \beta_1 \cdot \begin{cases} u_{\min}, & \beta_1 > 0 \\ u_{\max}, & \beta_1 < 0 \end{cases}}{L_0 + L_1 \cdot \begin{cases} \abs{u_{\min}}, & \beta_1 > 0 \\ u_{\max}, & \beta_1 < 0 \end{cases}}.		
\end{equation}		
or, respectively for $m = 2$,
\begin{equation}
	\bar \eps_1 = \max \left\{ \bar E_{012}, \bar E_{0i} \right\}
\end{equation}
with $\bar E_{012} := \min \left\{ E_1, E_2, E_{012} \right\}$, $\bar E_{0i} := \min \left\{ E_i, E_{0i} \right\}$, for those $i \in \set I_2$, where the index set
\begin{equation} \label{eqn:index-set}
	\set I_2 := \left\{ i \in \{1,2\}: \begin{cases} \beta_0 + \beta_i u_{i,\min} \leq 0, & \text{if } \beta_i > 0 \\ \beta_0 + \beta_i u_{i,\max} \leq 0, & \text{if } \beta_i < 0 \end{cases} \right\}
\end{equation}
and 
\begin{equation} \label{eqn:definition-E012}
	E_{012} = - \frac{\beta_0 + \sum \limits_{i = 1}^2 \beta_i \cdot \begin{cases} u_{i,\min}, & \beta_i > 0 \\ u_{i,\max}, & \beta_i < 0 \end{cases}}{L_0 + \sum \limits_{i = 1}^2 L_i \cdot \begin{cases} \abs{u_{i,\min}}, & \beta_i > 0 \\ u_{i,\max}, & \beta_i < 0 \end{cases}}.
\end{equation}
\end{document}